\documentclass[reqno,12pt]{amsart}

\usepackage{a4wide}
\usepackage[T1]{fontenc}
\usepackage{ae,aecompl}
\usepackage[english]{babel}
\usepackage[latin1]{inputenc}
\usepackage{enumerate}
\usepackage{latexsym}
\usepackage{amsmath,amsthm,amsfonts}
\usepackage{xr}
\usepackage{graphicx}
\usepackage{stmaryrd}
\usepackage{mathtools}
\usepackage{braket}

\newtheorem{satz}{Theorem}[section]

\newtheorem{lem}[satz]{Lemma}
\newtheorem{prop}[satz]{Proposition}
\newtheorem{defi}[satz]{Definition}
\newtheorem{bem}[satz]{Remark}

\newcommand{\sgn}{\operatorname{sgn}}

\renewcommand{\Re}{\operatorname{Re}}

\newcommand{\vvect}[2]{\left(\begin{array}{c} #1 \\ #2 \end{array}\right)}

\title[Schr\"odinger evolution of superoscillations with $\delta$- and $\delta'$-potentials]{\bf Schr\"odinger evolution of superoscillations with $\delta$- and $\delta'$-potentials}

\author[Yakir Aharonov]{Yakir Aharonov}
\address{(YA)
Schmid College of Science and Technology, Chapman University, Orange 92866, CA, USA}
\email{aharonov@chapman.edu}

\author[Jussi Behrndt]{Jussi Behrndt}
\address{(JB) Institut f\"ur Angewandte Mathematik, Technische Universit\"{a}t Graz, Steyrergasse 30, 8010 Graz, Austria}
\email{behrndt@tugraz.at}

\author[Fabrizio Colombo]{Fabrizio Colombo}
\address{(FC) Politecnico di Milano\\Dipartimento di Matematica\\Via E. Bonardi, 9\\20133 Milano, Italy}
\email{fabrizio.colombo@polimi.it}

\author[Peter Schlosser]{Peter Schlosser}
\address{(PS) Institut f\"ur Angewandte Mathematik, Technische Universit\"{a}t Graz, Steyrergasse 30, 8010 Graz, Austria}
\email{schlosser@tugraz.at}

\begin{document}

\begin{abstract}\vspace{1cm}
In this paper we study the time persistence of superoscillations as the initial data of the time dependent Schr\"odinger equation with $\delta$- and $\delta'$-potentials.
It is shown that the sequence of solutions converges uniformly on compact sets, whenever the initial data converges in the topology of the entire function space $A_1(\mathbb{C})$.
Convolution operators acting in this space are our main tool. In particular, a general result about the existence of such operators is proven.
Moreover, we provide an explicit formula as well as the large time asymptotics for the time evolution of a plane wave under $\delta$- and $\delta'$-potentials.
\end{abstract}

\maketitle

\par\noindent AMS Classification: 32A15, 32A10, 47B38.
\par\noindent \textit{Key words}: Superoscillating functions,  convolution operators, Schr\"odinger equation, singular potential, entire functions with growth conditions.

\vskip 1cm

\section{Introduction}

Superoscillating functions  have an oscillatory behaviour which is locally faster than their fastest Fourier component. This paradoxical property was discovered by the first author
and his collaborators in their work about weak measurements \cite{AAV88} and afterwards investigated from a mathematical and quantitative point of view by M.\,V. Berry \cite{B94}. In antenna theory
this phenomenon was discovered by G. Toraldo di Francia in  \cite{T52} as pointed out also in \cite{B14}.
Several authors have contributed to this field and without claiming completeness we mention \cite{AR05,AV90,BD09,BP06} and also \cite{FK04,FK06,LF14_1,LF14_2,K18,RRSYZ14,RYZ17}. More recently,
special emphasis was given to the mathematical aspect of superoscillations, see, e.g., \cite{ACNSST12,ACSST11,ACSST13,ACSST15,ACSST16,ACSST17_1,ACSST17_2,CGS17}.

\vspace{0.2cm}

The special topic we want to investigate in this paper is the Schr\"{o}dinger time evolution of superoscillating functions $F$, that is, we consider
\begin{equation}\label{Eq_Schroedinger}
\begin{split}
i\frac{\partial}{\partial t}\Psi(t,x)&=\left(-\frac{\partial^2}{\partial x^2}+V(x)\right)\Psi(t,x),\qquad t>0\,,x\in\mathbb{R}, \\
\Psi(0,x)&=F(x),\hspace{4.3cm} x\in\mathbb{R},
\end{split}
\end{equation}
where $V$ is the potential and $F$ is the initial datum, which is assumed to be superoscillating. Already a number of different potentials were investigated, see the survey papers \cite{ACST18,ASTY18,ACSS18,B19} and references therein.
Our aim is to add the one dimensional $\delta$-potential and $\delta'$-potential to this list.
We mention that the $\delta$-potential was already treated in \cite{BCS19}; however, a (technical) condition on the strength of the $\delta$-interaction was imposed there, which we
are able to avoid in the present paper.
A detailed discussion of Schr\"{o}dinger operators with $\delta$ and $\delta'$-point interactions can be found in the standard monograph \cite{AGHH05}. For further reading on Schr\"{o}dinger operators with singular potentials
we refer the reader to, e.g., \cite{BEL14,BLL13,B88,BEKS94,C09,EKMT14,EGU19,E08,EK15,ER16,H89,KM10,KM14,LO16,LR15,MPS16} and the references therein.

\vspace{0.2cm}

We briefly illustrate the concept of superoscillations. Consider for some fixed $k\in\mathbb{R}$ with $|k|>1$ the sequence of functions
\begin{equation}\label{Eq_Example_function}
F_n(z,k)=\sum\limits_{j=0}^nC_j(n,k)e^{i(1-\frac{2j}{n})z},\qquad z\in\mathbb{C},
\end{equation}
with coefficients
\begin{equation*}
C_j(n,k)=\vvect{n}{j}\left(\frac{1+k}{2}\right)^{n-j}\left(\frac{1-k}{2}\right)^j.
\end{equation*}
The notion {\it superoscillatory} now comes from the fact that, although all the Fourier coefficients $k_j(n)=1-\frac{2j}{n}$ are contained in the bounded interval $[-1,1]$, the whole sequence converges to
\begin{equation}\label{Eq_Example_convergence}
\lim\limits_{n\rightarrow\infty}F_n(z,k)=e^{ikz},
\end{equation}
a plane wave with wave vector $|k|>1$; cf. \cite[Theorem 2.1]{CSSY19} for more details. Besides the convergence (\ref{Eq_Example_convergence}), the important feature of the functions $F_n$ is that also for finite $n$ they oscillate
with frequencies close to $k$ in certain intervals; the length of these intervals grow when $n$ increases. Nevertheless, outside this
interval one obtains an exponential growth of the amplitude, which conversely means that the amplitude inside the superoscillatory region is exponentially small.
Different types of functions, in the form of a square-integrable $sinc$ function, which are band-limited and in some intervals oscillate faster than its highest Fourier component, can be found in \cite{B16}.

Inspired by (\ref{Eq_Example_function}), we define the notion of superoscillations as follows:

\begin{defi}\label{defi_Superoscillation}
A sequence of functions of the form
\begin{equation}\label{Eq_Generalized_Fourier_series}
F_n(z)=\sum\limits_{j=0}^nC_j(n)e^{ik_j(n)z},\qquad z\in\mathbb{C},
\end{equation}
with coefficients $C_j(n)\in\mathbb{C}$ and $k_j(n)\in\mathbb{R}$, is said to be superoscillating, if there exists some $k\in\mathbb{R}$ such that
\begin{enumerate}
\item[{\rm (i)}] $\sup_{n\in\mathbb{N}_0,j\in\{0,\dots,n\}}|k_j(n)|<k$, and
\item[{\rm (ii)}] for some $B\geq 0$ one has $\lim\limits_{n\rightarrow\infty}\big\Vert(F_n-e^{ik\,\cdot})e^{-B|\,\cdot\,|}\big\Vert_\infty=0$.
\end{enumerate}
\end{defi}

Note, that (ii) is exactly the convergence in the space $A_1(\mathbb{C})$ introduced in Definition~\ref{defi_A1} below.
Due to the above description of the exponential growth of the amplitude outside the superoscillatory region, it is reasonable to use the exponential weight $e^{-B|\,\cdot\,|}$ as a damping factor in the uniform convergence.

\vspace{0.2cm}

The purpose of this paper is now to consider a superoscillating sequence $(F_n)_n$ as the initial datum of the time dependent Schr\"odinger equation (\ref{Eq_Schroedinger}) with
either a $\delta$-potential or a $\delta'$-potential and investigate the corresponding sequence of solutions $(\Psi_n)_n$. The main result of this paper is the following:

\begin{satz}\label{satz_Convergence_solutions}
Let $(F_n)_n$ be a superoscillatory sequence with limit function
\begin{equation}\label{Eq_Convergence_initial_datum}
F_n\overset{A_1}{\longrightarrow}e^{ik\,\cdot\,},
\end{equation}
for some $k\in\mathbb{R}$. Then the solutions $\Psi$ and $\Psi_n$ of \eqref{Eq_Schroedinger} with either $V=2\alpha\delta$ or $V=\frac{2}{\beta}\delta'$ for $\alpha,\beta\in\mathbb{R}\setminus\{0\}$,
and initial data $e^{ik\,\cdot\,}$ and $F_n$, respectively, satisfy
\begin{equation*}
\lim\limits_{n\rightarrow\infty}\Psi_n(t,x)=\Psi(t,x)
\end{equation*}
uniformly on every compact subset of $(0,\infty)\times\mathbb{R}$.
\end{satz}

Note, that the mathematical rigorous implementation of a singular $\delta$-potential or a $\delta'$-potentials in the Schr\"odinger equation (\ref{Eq_Schroedinger}) is via jump
conditions at the location $x=0$ of the interaction; cf. (\ref{Eq_Schroedinger_delta}) and (\ref{Eq_Schroedinger_deltaprime}).

\vspace{0.2cm}

The proof of Theorem \ref{satz_Convergence_solutions} is postponed to the end of Section \ref{sec_Proof_of_Theorem}. Before, in Section \ref{sec_Preliminary_results}, we
introduce the space $A_1(\mathbb{C})$ which is used in the convergence (\ref{Eq_Convergence_initial_datum}). We also construct continuous operators acting in this space, which will play a
crucial role in the proof of Theorem \ref{satz_Convergence_solutions}. Moreover, in Section \ref{sec_plane_wave} we explicitly calculate the solution with a plane wave $F(x)=e^{ikx}$ as initial condition.

\section{Convolution operators in $A_1(\mathbb{C})$}\label{sec_Preliminary_results}

In this section we recall the definition of the space $A_1(\mathbb{C})$ of entire functions with exponential growth,
already mentioned below Definition \ref{defi_Superoscillation}. For the analysis of superoscillations this space (or slight modifications of it) is a convenient choice;
cf. \cite{ACSST17_1,ASTY18,ACSS18}. Note also
that $A_1(\mathbb{C})$ is one particular space in the theory of analytically uniform spaces, see, e.g., \cite[Chapter~4]{ACSST17_1}.

\begin{defi}\label{defi_A1}
Let $\mathcal{H}(\mathbb{C})$ be the space of entire functions and
define
\begin{equation*}
A_1(\mathbb{C})\coloneqq\Set{F\in\mathcal{H}(\mathbb{C}) | \exists A,B\geq 0\text{ such that }|F(z)|\leq Ae^{B|z|}\text{ for all }z\in\mathbb{C}}.
\end{equation*}
A sequence $(F_n)_n\in A_1(\mathbb{C})$ is said to be $A_1$-convergent to $F\in A_1(\mathbb{C})$ if
\begin{equation}\label{Eq_A1_convergence}
\lim\limits_{n\rightarrow\infty}\big\Vert(F_n-F)e^{-B|\,\cdot\,|}\big\Vert_\infty=0\quad\text{for some }B\geq 0.
\end{equation}
This type of convergence will be denoted by $F_n\overset{A_1}{\longrightarrow}F$.
\end{defi}

The following lemma shows that also the derivatives of functions in  $A_1(\mathbb{C})$ are exponentially bounded.

\begin{lem}
Let $F\in A_1(\mathbb{C})$ admit the estimate
\begin{equation}\label{Eq_A1_estimate}
|F(z)|\leq Ae^{B|z|},\qquad z\in\mathbb{C},
\end{equation}
for some $A,B\geq 0$. Then the derivatives $F^{(n)}$, $n\in\mathbb{N}$, of $F$ admit the estimate
\begin{equation}\label{Eq_A1_estimate_derivative}
|F^{(n)}(z)|\leq A(eB)^ne^{B|z|},\qquad z\in\mathbb{C},
\end{equation}
and, in particular, $F^{(n)}\in A_1(\mathbb{C})$ for all $n\in\mathbb{N}$.
\end{lem}

\begin{proof}
If $B=0$ the statement holds, since in that case the entire function $F$ is bounded and hence constant, so that all its derivatives vanish identically. In the following let
$B\neq 0$. By Cauchy's integral formula we have
\begin{equation*}
F^{(n)}(z)=\frac{n!}{2\pi i}\int_{|\xi-z|=r}\frac{F(\xi)}{(\xi-z)^{n+1}}d\xi=\frac{n!}{2\pi r^n}\int_0^{2\pi}\frac{F(z+re^{i\varphi})}{e^{in\varphi}}d\varphi,\qquad z\in\mathbb{C},
\end{equation*}
for the $n$-th derivative, where $r>0$ is not yet specified. The exponential boundedness (\ref{Eq_A1_estimate}) of the integrand gives the estimate
\begin{equation*}
\left|F^{(n)}(z)\right|\leq\frac{An!}{2\pi r^n}\int_0^{2\pi}e^{B|z+re^{i\varphi}|}d\varphi\leq\frac{An!}{r^n}e^{B(|z|+r)},\qquad z\in\mathbb{C}.
\end{equation*}
It is easy to see that the right hand side can be minimized choosing $r=\frac{n}{B}$, which gives
\begin{equation*}
\left|F^{(n)}(z)\right|\leq\frac{AB^nn!}{n^n}e^{B|z|+n},\qquad z\in\mathbb{C}.
\end{equation*}
Using $\frac{n!}{n^n}\leq 1$, this gives the estimate (\ref{Eq_A1_estimate_derivative}).
\end{proof}

After this preparatory lemma, we consider an operator that can be used to reconstruct a given function from plane waves.
Such operators will play the crucial role of time evolution operators in Section \ref{sec_Proof_of_Theorem}, acting on the initial datum of the Schr\"odinger equation and pointwise giving its solution.

\begin{prop}\label{prop_Operator_representation}
Let $\Psi:\mathbb{R}\rightarrow\mathbb{C}$ be a function, which can be written as the absolute convergent power series
\begin{equation}\label{Eq_Psi_power_series}
\Psi(k)=\sum\limits_{m=0}^\infty c_mk^m,\qquad k\in\mathbb{R},
\end{equation}
with coefficients $(c_m)_m\in\mathbb{C}$. Then the operator $U:A_1(\mathbb{C})\rightarrow A_1(\mathbb{C})$, defined by
\begin{equation}\label{Eq_Operator}
UF(\xi)\coloneqq\sum\limits_{m=0}^\infty(-i)^mc_m\frac{d^m}{d\xi^m}F(\xi),\qquad F\in A_1(\mathbb{C}),\,\xi\in\mathbb{C},
\end{equation}
is continuous in $A_1(\mathbb{C})$ and satisfies
\begin{equation}\label{Eq_Psi_Operator_representation}
\Psi(k)=Ue^{ik\xi}\big|_{\xi=0},\qquad k\in\mathbb{R}.
\end{equation}
Moreover, for $F\in A_1(\mathbb{C})$ such that $|F(\xi)|\leq Ae^{B|\xi|}$ and $S_B\coloneqq\sum\limits_{m=0}^\infty|c_m|B^m$  the estimate
\begin{equation}\label{Eq_Operator_boundedness}
|UF(\xi)|\leq AS_{eB}e^{eB|\xi|},\qquad\xi\in\mathbb{C},
\end{equation}
is valid.
\end{prop}

\begin{proof}
In order to see that for $F\in A_1(\mathbb{C})$ the image $UF$ is again an element in $A_1(\mathbb{C})$, we have to show that $UF$ is entire and exponentially bounded. Since $F\in A_1(\mathbb{C})$ is entire
we can use the power series representation
$$
F(\xi)=\sum_{n=0}^\infty\frac{F^{(n)}(0)}{n!}\xi^n,\qquad\xi\in\mathbb{C},
$$
to write (\ref{Eq_Operator}) as the power series
\[
\begin{split}
UF(\xi)&=\sum\limits_{m=0}^\infty(-i)^mc_m\sum\limits_{n=m}^\infty\frac{F^{(n)}(0)}{(n-m)!}\,\xi^{n-m}
\\
&
=\sum\limits_{n=0}^\infty\sum\limits_{m=0}^\infty(-i)^mc_m\frac{F^{(n+m)}(0)}{n!}\,\xi^n,\qquad\xi\in\mathbb{C}.
\end{split}
\]
In the last step we interchanged the order of summation, which is allowed since (\ref{Eq_A1_estimate_derivative}) ensures the estimate
\begin{equation}\label{Eq_Operator_existence_3}
\sum\limits_{m=0}^\infty|c_m|\frac{|F^{(n+m)}(0)|}{n!}\leq\sum\limits_{m=0}^\infty|c_m|\frac{A(eB)^{n+m}}{n!}=AS_{eB}\frac{(eB)^n}{n!},
\end{equation}
and hence the absolute convergence of the double sum. This shows that $UF$ can be represented as an everywhere convergent power series and hence is entire.
Moreover, the estimate (\ref{Eq_Operator_existence_3}) gives the exponential boundedness
\begin{equation*}
|UF(\xi)|\leq AS_{eB}\sum\limits_{n=0}^\infty\frac{(eB)^n}{n!}|\xi|^n=AS_{eB}e^{eB|\xi|},\qquad\xi\in\mathbb{C},
\end{equation*}
so that $UF\in A_1(\mathbb{C})$ and the estimate (\ref{Eq_Operator_boundedness}) is satisfied. Finally, (\ref{Eq_Psi_Operator_representation}) follows from
\begin{equation*}
Ue^{ik\xi}\big|_{\xi=0}=\sum\limits_{m=0}^\infty(-i)^mc_m\frac{d^m}{d\xi^m}e^{ik\xi}\Big|_{\xi=0}=\sum\limits_{m=0}^\infty(-i)^mc_m(ik)^m=\Psi(k).
\end{equation*}
\end{proof}

\section{Plane wave under a $\delta$- or $\delta'$-potential}\label{sec_plane_wave}

In quantum mechanics a particle or wave interacting with a $\delta$-potential or $\delta'$-potential is described by the Schr\"odinger equation (\ref{Eq_Schroedinger}),
formally putting $V=2\alpha\delta$ or $V=\frac{2}{\beta}\delta'$, where $\alpha,\beta\in\mathbb{R}\setminus\{0\}$ model the strength of the respective peaks.
The rigorous mathematical description of such singular potentials is via jump conditions at the location $x=0$ of the potential.
More precisely, for the $\delta$-potential we have to consider the system
\begin{subequations}\label{Eq_Schroedinger_delta}
\begin{align}
i\frac{\partial}{\partial t}\Psi_\delta(t,x)&=-\frac{\partial^2}{\partial x^2}\Psi_\delta(t,x),\qquad t>0,\,x\in\mathbb{R}\setminus\{0\}, \label{Eq_Schroedinger_delta_1} \\
\Psi_\delta(t,0^+)&=\Psi_\delta(t,0^-),\hspace{1.66cm} t>0, \label{Eq_Schroedinger_delta_2} \\
\frac{\partial}{\partial x}\Psi_\delta(t,0^+)-\frac{\partial}{\partial x}\Psi_\delta(t,0^-)&=2\alpha\,\Psi_\delta(t,0),\hspace{1.37cm} t>0, \label{Eq_Schroedinger_delta_3} \\
\Psi_\delta(0^+,x)&=F(x),\hspace{2.37cm} x\in\mathbb{R}\setminus\{0\}, \label{Eq_Schroedinger_delta_4}
\end{align}
\end{subequations}
and for the $\delta'$-potential we have to consider the system
\begin{subequations}\label{Eq_Schroedinger_deltaprime}
\begin{align}
i\frac{\partial}{\partial t}\Psi_{\delta'}(t,x)&=-\frac{\partial^2}{\partial x^2}\Psi_{\delta'}(t,x),\qquad t>0,\,x\in\mathbb{R}\setminus\{0\}, \label{Eq_Schroedinger_deltaprime_1} \\
\frac{\partial}{\partial x}\Psi_{\delta'}(t,0^+)&=\frac{\partial}{\partial x}\Psi_{\delta'}(t,0^-),\hspace{1.1cm} t>0, \label{Eq_Schroedinger_deltaprime_2} \\
\Psi_{\delta'}(t,0^+)-\Psi_{\delta'}(t,0^-)&=\frac{2}{\beta}\frac{\partial}{\partial x}\Psi_{\delta'}(t,0),\hspace{1cm} t>0, \label{Eq_Schroedinger_deltaprime_3} \\
\Psi_{\delta'}(0^+,x)&=F(x),\hspace{2.47cm} x\in\mathbb{R}\setminus\{0\}, \label{Eq_Schroedinger_deltaprime_4}
\end{align}
\end{subequations}
where $F$ is an appropriate initial condition. In the case that $F$ is a plane wave we find an explicit representation of the solutions $\Psi_\delta$ and $\Psi_{\delta'}$
in Theorem \ref{satz_Plane_wave_solution} below. To write down these solutions efficiently we will use a certain modified error function $\Lambda$ and its properties.
For the convenience of the reader we recall the following lemma from \cite{BCS19}.

\begin{lem}\label{lem_Lambda}
For the function
\begin{equation*}
\Lambda(z)\coloneqq e^{z^2}\left(1-\frac{2}{\sqrt{\pi}}\int_0^ze^{-\xi^2}d\xi\right),\qquad z\in\mathbb{C},
\end{equation*}
the following statements hold:

\begin{enumerate}
\item[{\rm (i)}] $\Lambda(-z)=2e^{z^2}-\Lambda(z)$, $z\in\mathbb{C}$;

\item[{\rm (ii)}] $\Lambda'(z)=2z\Lambda(z)-\frac{2}{\sqrt{\pi}}$, $z\in\mathbb{C}$;

\item[{\rm (iii)}] for $|z|\rightarrow\infty$ one has
\begin{equation}\label{Eq_Lambda_asymptotics}
\Lambda(z)=\left\{\begin{array}{ll} \frac{1}{\sqrt{\pi}z}+\mathcal{O}\big(\frac{1}{|z|^2}\big),
& \text{if }\Re(z)\geq 0, \\ 2e^{z^2}+\frac{1}{\sqrt{\pi}z}+\mathcal{O}\big(\frac{1}{|z|^2}\big), & \text{if }\Re(z)\leq 0; \end{array}\right.
\end{equation}

\item[{\rm (iv)}] $\Lambda$ admits the power series representation
\begin{equation}\label{Eq_Lambda_power_series}
\Lambda(z)=\sum\limits_{n=0}^\infty\frac{(-1)^n}{\Gamma(\frac{n}{2}+1)}z^n,\qquad z\in\mathbb{C}.
\end{equation}
\end{enumerate}
\end{lem}

The next theorem is the main result in this section.

\begin{satz}\label{satz_Plane_wave_solution}
The Schr\"odinger equations \eqref{Eq_Schroedinger_delta} and \eqref{Eq_Schroedinger_deltaprime} with initial datum $F(x)=e^{ikx}$, $k\in\mathbb{R}$, admit the explicit solutions
\begin{subequations}\label{Eq_delta_deltaprime_solutions}
\begin{align}
\Psi_\delta(t,x;k)&=\Psi_\text{\rm free}(t,x;k)+\varphi_\delta(t,x;k)+\varphi_\delta(t,-x;-k), \label{Eq_delta_solution} \\
\Psi_{\delta'}(t,x;k)&=\Psi_\text{\rm free}(t,x;k)+\varphi_{\delta'}(t,x;k)+\varphi_{\delta'}(t,-x;-k), \label{Eq_deltaprime_solution}
\end{align}
\end{subequations}
where $\Psi_\text{\rm free}(t,x;k)=e^{ikx-ik^2t}$ is the solution of the unperturbed system and
\begin{subequations}
\begin{align}
\varphi_\delta(t,x;k)&=\frac{\alpha}{2(\alpha+ik)}e^{-\frac{x^2}{4it}}
\left(\Lambda\left(\frac{|x|}{2\sqrt{it}}+\alpha\sqrt{it}\right)-\Lambda\left(\frac{|x|}{2\sqrt{it}}-ik\sqrt{it}\right)\right), \label{Eq_varphi_delta} \\
\varphi_{\delta'}(t,x;k)&=\frac{\sgn(x)}{2(\beta+ik)}e^{-\frac{x^2}{4it}}\left(
\beta\Lambda\left(\frac{|x|}{2\sqrt{it}}+\beta\sqrt{it}\right)+ik\Lambda\left(\frac{|x|}{2\sqrt{it}}-ik\sqrt{it}\right)\right). \label{Eq_varphi_deltaprime}
\end{align}
\end{subequations}
Moreover, for $t\rightarrow\infty$ the wave functions \eqref{Eq_delta_deltaprime_solutions} satisfy
\begin{subequations}\label{Eq_asymptotics}
\begin{align}
\Psi_\delta(t,x;k)&=e^{-ik^2t}\left(e^{ikx}-\frac{\alpha}{\alpha-i|k|}e^{i|kx|}\right)+1_{\mathbb{R}^-}(\alpha)\frac{2\alpha^2}{\alpha^2+k^2}e^{\alpha|x|+i\alpha^2t}+\mathcal{O}\left(\frac{1}{t}\right), \label{Eq_delta_asymptotics} \\
\Psi_{\delta'}(t,x;k)&=e^{-ik^2t}\left(e^{ikx}+\frac{ik\sgn(x)}{\beta-i|k|}e^{i|kx|}\right)-1_{\mathbb{R}^-}(\beta)\frac{2i\beta k\sgn(x)}{\beta^2+k^2}e^{\beta|x|+i\beta^2t}+\mathcal{O}\left(\frac{1}{t}\right), \label{Eq_deltaprime_asymptotics}
\end{align}
\end{subequations}
where $1_{\mathbb{R}^-}$
denotes the characteristic function of the negative half line.
\end{satz}
\begin{bem}{\rm
Note that the second terms on the right hand sides in the asymptotics (\ref{Eq_asymptotics}) only appear for the attractive cases $\alpha,\beta<0$.
These terms are connected to the negative bound states which are only present in the attractive cases.
More precisely, these terms can be seen as the time evolution of the eigenfunctions $e^{\alpha|x|}$ and $\sgn(x)e^{\beta|x|}$ corresponding to the eigenvalues
$-\alpha^2$ and $-\beta^2$, respectively.
}\end{bem}

\begin{proof}[Proof of Theorem \ref{satz_Plane_wave_solution}]
Since the strategy of the proof is the same for $\Psi_\delta$ and $\Psi_{\delta'}$, and since the $\delta$-potential is already treated in \cite{BCS19},
we restrict our considerations to the $\delta'$-case. The verification of (\ref{Eq_Schroedinger_deltaprime}) and (\ref{Eq_deltaprime_asymptotics}) is splitted into four steps.

\vspace{0.2cm}\noindent
\textit{Step 1.} We check that (\ref{Eq_deltaprime_solution}) satisfies the differential equation (\ref{Eq_Schroedinger_deltaprime_1}).
Since this is obvious for the free part $\Psi_\text{\rm free}$ it suffices to check this for (\ref{Eq_varphi_deltaprime}). In fact, it suffices to verify that
\begin{equation}\label{Eq_phi}
\phi(t,x;\omega)\coloneqq e^{-\frac{x^2}{4it}}\Lambda\left(\frac{|x|}{2\sqrt{it}}+\omega\sqrt{it}\right),\qquad t>0,\,x\in\mathbb{R}\setminus\{0\},
\end{equation}
is a solution of (\ref{Eq_Schroedinger_deltaprime_1}) for every $\omega\in\mathbb{C}$. With the help of Lemma~\ref{lem_Lambda}~(ii) we first compute
\begin{align}
\frac{\partial}{\partial t}\phi(t,x;\omega)&=ie^{-\frac{x^2}{4it}}\left(\omega^2\Lambda\left(\frac{|x|}{2\sqrt{it}}+\omega\sqrt{it}\right)+\frac{1}{it\sqrt{\pi}}\left(\frac{|x|}{2\sqrt{it}}-\omega\sqrt{it}\right)\right), \notag\\
\frac{\partial}{\partial x}\phi(t,x;\omega)&=\sgn(x)e^{-\frac{x^2}{4it}}\left(\omega\Lambda\left(\frac{|x|}{2\sqrt{it}}+\omega\sqrt{it}\right)-\frac{1}{\sqrt{i\pi t}}\right), \label{Eq_phi_x} \\
\frac{\partial^2}{\partial x^2}\phi(t,x;\omega)&=e^{-\frac{x^2}{4it}}\left(\omega^2\Lambda\left(\frac{|x|}{2\sqrt{it}}+\omega\sqrt{it}\right)+\frac{1}{it\sqrt{\pi}}\left(\frac{|x|}{2\sqrt{it}}-\omega\sqrt{it}\right)\right). \notag
\end{align}
Using this with $\omega=\beta$ and $\omega=-ik$ immediately shows that $\varphi_{\delta'}$ and hence $\Psi_{\delta'}$ solve (\ref{Eq_Schroedinger_deltaprime_1}).

\vspace{0.2cm}
\noindent
\textit{Step 2.} We verify that the jump conditions (\ref{Eq_Schroedinger_deltaprime_2}) and (\ref{Eq_Schroedinger_deltaprime_3}) are satisfied. For this
we calculate the limits $x\rightarrow 0^\pm$ of (\ref{Eq_phi}) and (\ref{Eq_phi_x}) and insert them into (\ref{Eq_varphi_deltaprime}), to get
\begin{align*}
\varphi_{\delta'}(t,0^\pm;k)&=\pm\frac{\beta\Lambda\left(\beta\sqrt{it}\right)+ik\Lambda\left(-ik\sqrt{it}\right)}{2(\beta+ik)}, \\
\frac{\partial}{\partial x}\varphi_{\delta'}(t,0^\pm;k)&=\frac{\beta^2\Lambda\left(\beta\sqrt{it}\right)+k^2\Lambda\left(-ik\sqrt{it}\right)}{2(\beta+ik)}-\frac{1}{2\sqrt{i\pi t}}.
\end{align*}
Plugging this into (\ref{Eq_deltaprime_solution}) gives
\begin{align*}
\Psi_{\delta'}(t,x;k)&=e^{-ik^2t}\mp\frac{ik\beta}{\beta^2+k^2}\Lambda\left(\beta\sqrt{it}\right)\pm\frac{ik}{2}\left(\frac{\Lambda\left(-ik\sqrt{it}\right)}{\beta+ik}+\frac{\Lambda\left(ik\sqrt{it}\right)}{\beta-ik}\right), \\
\frac{\partial}{\partial x}\Psi_{\delta'}(t,x;k)&=ike^{-ik^2t}-\frac{ik\beta^2}{\beta^2+k^2}\Lambda\left(\beta\sqrt{it}\right)+\frac{k^2}{2}\left(\frac{\Lambda\left(-ik\sqrt{it}\right)}{\beta+ik}-\frac{\Lambda\left(ik\sqrt{it}\right)}{\beta-ik}\right).
\end{align*}
Finally, using property (i) in Lemma~\ref{lem_Lambda} of $\Lambda(-z)$ leads to
\begin{align*}
\Psi_{\delta'}(t,0^\pm;k)&=e^{-ik^2t}\left(1\pm\frac{ik}{\beta+ik}\right)\mp\frac{ik\beta}{\beta^2+k^2}\Lambda\left(\beta\sqrt{it}\right)\mp\frac{k^2}{\beta^2+k^2}\Lambda\left(ik\sqrt{it}\right), \\
\frac{\partial}{\partial x}\Psi_{\delta'}(t,0^\pm;k)&=\beta\left(\frac{ik}{\beta+ik}e^{-ik^2t}-\frac{ik\beta}{\beta^2+k^2}\Lambda\left(\beta\sqrt{it}\right)-\frac{k^2}{\beta^2+k^2}\Lambda\left(ik\sqrt{it}\right)\right),
\end{align*}
from which it is clear that (\ref{Eq_Schroedinger_deltaprime_2}) and (\ref{Eq_Schroedinger_deltaprime_3}) are both satisfied.

\vspace{0.2cm}
\noindent
\textit{Step 3.} We check the initial condition (\ref{Eq_Schroedinger_deltaprime_4}). In fact, due to the asymptotic behaviour (\ref{Eq_Lambda_asymptotics})
we have $\varphi_{\delta'}(t,x;k)\rightarrow 0$ for $t\rightarrow 0^+$, and hence
\begin{equation*}
\Psi_{\delta'}(0^+,x;k)=\Psi_\text{free}(0^+,x;k)=e^{ikx}.
\end{equation*}

\noindent
\textit{Step 4.} In this step we show the large time asymptotics (\ref{Eq_deltaprime_asymptotics}) of the solution. Note first that for large $t$ we have
\begin{align*}
\Re\left(\frac{|x|}{2\sqrt{it}}+\beta\sqrt{it}\right)\geq 0,\qquad\text{if }\beta>0, \\
\Re\left(\frac{|x|}{2\sqrt{it}}+\beta\sqrt{it}\right)\leq 0,\qquad\text{if }\beta<0.
\end{align*}
Hence, from (\ref{Eq_Lambda_asymptotics}) we get for every fixed $x\in\mathbb{R}\setminus\{0\}$ the asymptotic behaviour
\begin{align*}
\Lambda\left(\frac{|x|}{2\sqrt{it}}+\beta\sqrt{it}\right)&=1_{\mathbb{R}^-}(\beta)2e^{\left(\frac{|x|}{2\sqrt{it}}+\beta\sqrt{it}\right)^2}+\frac{1}{\sqrt{\pi}\left(\frac{|x|}{2\sqrt{it}}+\beta\sqrt{it}\right)}+\mathcal{O}\left(\frac{1}{\left|\frac{|x|}{2\sqrt{it}}+\beta\sqrt{it}\right|^2}\right) \\
&=1_{\mathbb{R}^-}(\beta)2e^{\left(\frac{|x|}{2\sqrt{it}}+\beta\sqrt{it}\right)^2}+\frac{2\sqrt{it}}{\sqrt{\pi}\left(|x|+2\beta it\right)}+\mathcal{O}\left(\frac{1}{t}\right).
\end{align*}
Similarly, we get the asymptotics
\begin{equation*}
\Lambda\left(\frac{|x|}{2\sqrt{it}}-ik\sqrt{it}\right)=1_{\mathbb{R}^-}(k)2e^{\left(\frac{|x|}{2\sqrt{it}}-ik\sqrt{it}\right)^2}+
\frac{2\sqrt{it}}{\sqrt{\pi}\left(|x|+2kt\right)}+\mathcal{O}\left(\frac{1}{t}\right)
\end{equation*}
for the second term in (\ref{Eq_varphi_deltaprime}). Substituting this into (\ref{Eq_varphi_deltaprime}) gives
\begin{align*}
\varphi_{\delta'}(t,x;k)=\sgn(x)&\left(\frac{1_{\mathbb{R}^-}(\beta)\beta}{\beta+ik}e^{\beta|x|+i\beta^2t}+\frac{1_{\mathbb{R}^-}(k)ik}{\beta+ik}e^{-ik|x|-ik^2t}\right. \\
&\hspace{3.45cm}\left.+\frac{|x|\sqrt{it}}{\sqrt{\pi}(|x|+2\beta it)(|x|+2kt)}\right)+\mathcal{O}\left(\frac{1}{t}\right).
\end{align*}
Since the third summand in this expression is of order $t^{-3/2}$ and hence, in particular, $\mathcal{O}(\frac{1}{t})$, the expansion reduces to
\begin{equation*}
\varphi_{\delta'}(t,x;k)=\sgn(x)\left(\frac{1_{\mathbb{R}^-}(\beta)\beta}{\beta+ik}e^{\beta|x|+i\beta^2t}+\frac{1_{\mathbb{R}^-}(k)ik}{\beta+ik}e^{-ik|x|-ik^2t}\right)+\mathcal{O}\left(\frac{1}{t}\right).
\end{equation*}
Using this in (\ref{Eq_deltaprime_solution}), immediately gives the asymptotics (\ref{Eq_deltaprime_asymptotics}) of $\Psi_{\delta'}$.
\end{proof}

\section{Proof of Theorem \ref{satz_Convergence_solutions}}\label{sec_Proof_of_Theorem}

For the proof of Theorem \ref{satz_Convergence_solutions} some preparatory statements are needed.
The following lemma  provides
a simple algebraic reformulation of the power series of the difference quotient of an entire function.

\begin{lem}\label{lem_Series_expansion_of_the_difference_quotient}
Let $F\in\mathcal{H}(\mathbb{C})$ be an entire function, that is, $F$ admits the power series representation
\begin{equation}\label{Eq_Difference_quotient_power_series}
F(z)=\sum\limits_{n=0}^\infty f_nz^n,\qquad z\in\mathbb{C},
\end{equation}
with the coefficients $f_n=\frac{F^{(n)}(0)}{n!}$.
Then for every $a\in\mathbb{C}\setminus\{0\}$ the difference quotient admits the series representation
\begin{equation*}
\frac{F(z+a)-F(z)}{a}=\sum\limits_{m=0}^\infty\sum\limits_{n=0}^\infty f_{n+m+1}\vvect{n+m+1}{m+1}z^na^m,\qquad z\in\mathbb{C}.
\end{equation*}
\end{lem}

\begin{proof}
Inserting the power series (\ref{Eq_Difference_quotient_power_series}) into the difference quotient gives
\begin{equation}\label{Eq_Difference_quotient_1}
\frac{F(z+a)-F(z)}{a}=\sum\limits_{n=1}^\infty f_n\frac{(z+a)^n-z^n}{a}.
\end{equation}
We can now use the binomic formula
\begin{equation*}
(z+a)^n=z^n+\sum\limits_{m=1}^n\vvect{n}{m}z^{n-m}a^m
\end{equation*}
to rewrite the series (\ref{Eq_Difference_quotient_1}) as
\begin{align*}
\frac{F(z+a)-F(z)}{a}&=\sum\limits_{n=1}^\infty\frac{f_n}{a}\sum\limits_{m=1}^n\vvect{n}{m}z^{n-m}a^m
\\
&
=\sum\limits_{m=1}^\infty\sum\limits_{n=m}^\infty\frac{f_n}{a}\vvect{n}{m}z^{n-m}a^m \\
&=\sum\limits_{m=0}^\infty\sum\limits_{n=0}^\infty f_{n+m+1}\vvect{n+m+1}{m+1}z^na^m.
\end{align*}
\end{proof}

The following Lemma \ref{lem_Existence_Propagator_delta} for the $\delta$-potential
and its counterpart Lemma \ref{lem_Existence_Propagator_deltaprime} for the $\delta'$-potential are further useful ingredients in the proof of
Theorem \ref{satz_Convergence_solutions}.
Here Lemma \ref{lem_Series_expansion_of_the_difference_quotient} is used to construct time evolution operators,
which take the initial value of the Schr\"odinger equation and pointwise give the solutions (\ref{Eq_delta_deltaprime_solutions}).
In the case of a $\delta$-potential a similar operator was already provided in \cite{BCS19}; however, there
the (technical) restriction $|k|<|\alpha|$ (wave vector smaller than the potential strength) appeared, which is avoided here.

\begin{lem}\label{lem_Existence_Propagator_delta}
For any fixed $x\in\mathbb{R}\setminus\{0\}$ and $t>0$ there exists a continuous linear operator
$U_\delta(t,x):A_1(\mathbb{C})\rightarrow A_1(\mathbb{C})$ such that \eqref{Eq_delta_solution} can be represented as
\begin{equation}\label{Eq_Existence_Propagator_delta}
\Psi_\delta(t,x;k)=U_\delta(t,x)e^{ik\xi}\big|_{\xi=0}.
\end{equation}
Furthermore, for every $B\geq 0$ there exists $\widetilde{S}_{B,\delta}(t,x)$, continuous in $t$ and $x$, such that the estimate
\begin{equation}\label{Eq_Estimate_propagator_delta}
|U_\delta(t,x)F(\xi)|\leq A\widetilde{S}_{eB,\delta}(t,x)e^{eB|\xi|}
\end{equation}
holds for every $F\in A_1(\mathbb{C})$ satisfying $|F(\xi)|\leq Ae^{B|\xi|}$.
\end{lem}

\begin{proof}
Since $\Psi_\delta$ decomposes into (\ref{Eq_delta_solution}) it is sufficient to ensure (\ref{Eq_Existence_Propagator_delta})
and (\ref{Eq_Estimate_propagator_delta}) for its components $\Psi_\text{free}$ and $\varphi_\delta$. Starting with $\Psi_\text{free}$ we can write it as the power series
\begin{equation}\label{Eq_Psifree_series}
\begin{split}
\Psi_\text{free}(t,x;k)&=\sum\limits_{m=0}^\infty\frac{1}{m!}(ikx-ik^2t)^m\\
&=\sum\limits_{m=0}^\infty\frac{1}{m!}\sum_{j=0}^m\vvect{m}{j}(ikx)^{m-j}(-ik^2t)^j\\
&=\sum\limits_{m=0}^\infty\frac{1}{m!}\sum_{j=m}^{2m}\vvect{m}{j-m}(ix)^{2m-j}(-it)^{j-m} k^j\\
&=\sum\limits_{j=0}^\infty\sum\limits_{m=\lceil\frac{j}{2}\rceil}^j\frac{(ix)^{2m-j}(-it)^{j-m}}{m!}\vvect{m}{j-m}k^j,
\end{split}
\end{equation}
which is exactly the form (\ref{Eq_Psi_power_series}). Hence, by Proposition \ref{prop_Operator_representation} there exists
a continuous operator $U_\text{free}(t,x):A_1(\mathbb{C})\rightarrow A_1(\mathbb{C})$ such that
\begin{equation}\label{Eq_Psifree_operator}
\Psi_\text{free}(t,x;k)=U_\text{free}(t,x)e^{ik\xi}\big|_{\xi=0}.
\end{equation}
Argueing in the same way as in \eqref{Eq_Psifree_series} the respective constant $S_{B,\text{free}}(t,x)$ in the bound (\ref{Eq_Operator_boundedness}) can be estimated by
\begin{equation}\label{Eq_Free_constant}
\begin{split}
S_{B,\text{free}}(t,x)&\leq\sum\limits_{j=0}^\infty\sum\limits_{m=\lceil\frac{j}{2}\rceil}^j\frac{|x|^{2m-j}t^{j-m}}{m!}\vvect{m}{j-m}B^j \\
&=\sum\limits_{m=0}^\infty\frac{1}{m!}(B\vert x\vert + B^2t)^m\\
&=e^{B|x|+B^2t}\eqqcolon\widetilde{S}_{B,\text{free}}(t,x).
\end{split}
\end{equation}
Next we rearrange the terms of the function $\varphi_\delta$ in (\ref{Eq_varphi_delta}) in the form
\begin{equation}\label{Eq_varphi_difference_quotient}
\varphi_\delta(t,x;k)=\frac{\alpha\sqrt{it}}{2}e^{-\frac{x^2}{4it}}\frac{\Lambda\left(\frac{|x|}{2\sqrt{it}}+\alpha\sqrt{it}\right)-
\Lambda\left(\frac{|x|}{2\sqrt{it}}-ik\sqrt{it}\right)}{(\alpha+ik)\sqrt{it}}.
\end{equation}
It follows that this is a difference quotient of the form (\ref{Eq_Difference_quotient_power_series}) with $z=\frac{|x|}{2\sqrt{it}}-ik\sqrt{it}$ and $a=(\alpha+ik)\sqrt{it}$. Hence we obtain the series expansion \begin{equation*}
\varphi_\delta(t,x;k)=\frac{\alpha e^{-\frac{x^2}{4it}}}{2}\sum\limits_{m=0}^\infty\sum\limits_{n=0}^\infty\frac{(-\sqrt{it})^{n+m+1}}{\Gamma(\frac{n+m+3}{2})}\vvect{n+m+1}{m+1}\left(\frac{|x|}{2it}-ik\right)^n\left(\alpha+ik\right)^m,
\end{equation*}
where we used the coefficients in the power series representation (\ref{Eq_Lambda_power_series}) of $\Lambda$.
In a similar way as in (\ref{Eq_Psifree_series}) and (\ref{Eq_Free_constant}) for the free part,
we can use the binomic formula to further expand this series into the form (\ref{Eq_Psi_power_series}).
Hence we get an operator $U_{\delta,+}(t,x):A_1(\mathbb{C})\rightarrow A_1(\mathbb{C})$ satisfying
\begin{equation*}
\varphi_\delta(t,x;k)=U_{\delta,+}(t,x)e^{ik\xi}\big|_{\xi=0},
\end{equation*}
with a constant $S_{B,\delta,+}(t,x)$ in the corresponding bound (\ref{Eq_Operator_boundedness}). This constant can be estimated by
\begin{align*}
S_{B,\delta,+}(t,x)&\leq\frac{|\alpha|}{2}\sum\limits_{m=0}^\infty\sum\limits_{n=0}^\infty\frac{(\sqrt{t})^{n+m+1}}{\Gamma(\frac{n+m+3}{2})}\vvect{n+m+1}{m+1}\left(\frac{|x|}{2t}+B\right)^n\left(|\alpha|+B\right)^m \\
&=\frac{|\alpha|}{2}\frac{\Lambda\left(-\frac{|x|}{2\sqrt{t}}-(2B+|\alpha|)\sqrt{t}\right)
-\Lambda\left(-\frac{|x|}{2\sqrt{t}}-B\sqrt{t}\right)}{(|\alpha|+B)\sqrt{t}}\eqqcolon\widetilde{S}_{B,\delta,+}(t,x).
\end{align*}
The same reasoning for $\varphi_\delta(t,-x;-k)$ leads to an  operator $U_{\delta,-}(t,x):A_1(\mathbb{C})\rightarrow A_1(\mathbb{C})$ and a constant
$\widetilde{S}_{B,\delta,-}(t,x)$. Finally, we add all three terms together to get
\begin{align*}
U_\delta(t,x)&\coloneqq U_{\text{free}}(t,x)+U_{\delta,+}(t,x)+U_{\delta,-}(t,x), \\
\widetilde{S}_{B,\delta}(t,x)&\coloneqq\widetilde{S}_{B,\text{free}}(t,x)+\widetilde{S}_{B,\delta,+}(t,x)+\widetilde{S}_{B,\delta,-}(t,x),
\end{align*}
which satisfy (\ref{Eq_Existence_Propagator_delta}) and (\ref{Eq_Estimate_propagator_delta}).
\end{proof}

Although the proof is slightly different an analogous statement as in Lemma~\ref{lem_Existence_Propagator_delta} holds for the $\delta'$-potential.

\begin{lem}\label{lem_Existence_Propagator_deltaprime}
For any fixed $x\in\mathbb{R}\setminus\{0\}$ and $t>0$ there exists a continuous linear operator
$U_{\delta'}(t,x):A_1(\mathbb{C})\rightarrow A_1(\mathbb{C})$ such that \eqref{Eq_deltaprime_solution} can be represented as
\begin{equation}\label{Eq_Existence_Propagator_deltaprime}
\Psi_{\delta'}(t,x;k)=U_{\delta'}(t,x)e^{ik\xi}\big|_{\xi=0}.
\end{equation}
Furthermore, for every $B\geq 0$ there exists $\widetilde{S}_{B,\delta'}(t,x)$, continuous in $t$ and $x$, such that the estimate
\begin{equation}\label{Eq_Estimate_propagator_deltaprime}
|U_{\delta'}(t,x)F(\xi)|\leq A\widetilde{S}_{eB,\delta'}(t,x)e^{eB|\xi|}
\end{equation}
holds for every $F\in A_1(\mathbb{C})$ satisfying $|F(\xi)|\leq Ae^{B|\xi|}$.
\end{lem}

\begin{proof}
Since $\Psi_{\delta'}$ decomposes into (\ref{Eq_deltaprime_solution}) it is sufficient to ensure
(\ref{Eq_Existence_Propagator_deltaprime}) and (\ref{Eq_Estimate_propagator_deltaprime}) for its components $\Psi_\text{free}$ and $\varphi_{\delta'}$.
Since $\Psi_\text{free}$ is already treated in (\ref{Eq_Psifree_operator}) and (\ref{Eq_Free_constant}) we only consider $\varphi_{\delta'}$.
For this we first split (\ref{Eq_varphi_deltaprime}) into $\varphi_{\delta'}=\varphi_{\delta'}^{(0)}+\varphi_{\delta'}^{(1)}$, where
\begin{subequations}
\begin{align}
\varphi_{\delta'}^{(0)}(t,x;k)&\coloneqq\frac{\sgn(x)\beta\sqrt{it}}{2}e^{-\frac{x^2}{4it}}\frac{\Lambda\left(\frac{|x|}{2\sqrt{it}}+\beta\sqrt{it}\right)-\Lambda\left(\frac{|x|}{2\sqrt{it}}-ik\sqrt{it}\right)}{(\beta+ik)\sqrt{it}}, \label{Eq_varphi_deltaprime_0} \\
\varphi_{\delta'}^{(1)}(t,x;k)&\coloneqq\frac{\sgn(x)}{2}e^{-\frac{x^2}{4it}}\Lambda\left(\frac{|x|}{2\sqrt{it}}-ik\sqrt{it}\right). \label{Eq_varphi_deltaprime_1}
\end{align}
\end{subequations}
By comparing (\ref{Eq_varphi_deltaprime_0}) with (\ref{Eq_varphi_difference_quotient}) we see that, besides the presence of the prefactor $\sgn(x)$,
both coincide if  $\alpha$ is replaced by $\beta$. Hence the same computations as in then proof of Lemma~\ref{lem_Existence_Propagator_delta} lead to
operators $U_{\delta',\pm}^{(0)}(t,x)$ and corresponding constants $\widetilde{S}_{B,\delta',\pm}^{(0)}(t,x)$ for $\varphi_{\delta'}^{(0)}(t,\pm x;\pm k)$.

For the treatment of (\ref{Eq_varphi_deltaprime_1}) we first use (\ref{Eq_Lambda_power_series}) to rewrite $\varphi_{\delta'}^{(1)}(t,x;k)$ as the series
\begin{equation*}
\varphi_{\delta'}^{(1)}(t,x;k)=\frac{\sgn(x)e^{-\frac{x^2}{4it}}}{2}\sum\limits_{n=0}^\infty\frac{(-1)^n}{\Gamma(\frac{n}{2}+1)}\left(\frac{|x|}{2\sqrt{it}}-ik\sqrt{it}\right)^n.
\end{equation*}
Proposition \ref{prop_Operator_representation} then gives a continuous operator $U_{\delta',+}^{(1)}(t,x):A_1(\mathbb{C})\rightarrow A_1(\mathbb{C})$ such that
\begin{equation*}
\varphi_{\delta'}^{(1)}(t,x;k)=U_{\delta',+}^{(1)}(t,x)e^{ik\xi}\big|_{\xi=0},
\end{equation*}
as well as a corresponding constant $S_{B,\delta',+}^{(1)}(t,x)$ in the bound  (\ref{Eq_Operator_boundedness}). The constant can be estimated by
\[
\begin{split}
S_{B,\delta',+}^{(1)}(t,x)&\leq\frac{1}{2}\sum\limits_{n=0}^\infty\frac{1}{\Gamma(\frac{n}{2}+1)}\left(\frac{|x|}{2\sqrt{t}}+B\sqrt{t}\right)^n
\\
&
=\frac{1}{2}\Lambda\left(-\frac{|x|}{2\sqrt{t}}-B\sqrt{t}\right)\eqqcolon\widetilde{S}_{B,\delta',+}^{(1)}(t,x).
\end{split}
\]
In the same way we get $U_{\delta',-}^{(1)}(t,x):A_1(\mathbb{C})\rightarrow A_1(\mathbb{C})$ and $\widetilde{S}_{B,\delta',-}^{(1)}$ for $\varphi_{\delta'}(t,-x;-k)$.
Now we sum up all terms to end up with
\begin{align*}
U_{\delta'}(t,x)&\coloneqq U_\text{free}(t,x)+U_{\delta',+}^{(0)}(t,x)+U_{\delta',-}^{(0)}(t,x)+U_{\delta',+}^{(1)}(t,x)+U_{\delta',-}^{(1)}(t,x), \\
\widetilde{S}_{B,\delta'}(t,x)&\coloneqq\widetilde{S}_{B,\text{free}}(t,x)+\widetilde{S}_{B,\delta',+}^{(0)}(t,x)+\widetilde{S}_{B,\delta',-}^{(0)}(t,x)+\widetilde{S}_{B,\delta',+}^{(1)}(t,x)+\widetilde{S}_{B,\delta',-}^{(1)}(t,x),
\end{align*}
and it follows that (\ref{Eq_Existence_Propagator_deltaprime}) and (\ref{Eq_Estimate_propagator_deltaprime}) are satisfied.
\end{proof}

Now we are ready to prove Theorem \ref{satz_Convergence_solutions}.
Since the argument is the same for the $\delta$-potential and the $\delta'$-potential, we will only consider the $\delta'$-case. From Lemma \ref{lem_Existence_Propagator_deltaprime} we know that the
solution $\Psi_{\delta'}$ of the Schr\"odinger equations (\ref{Eq_Schroedinger_deltaprime}) with initial datum $F(x)=e^{ikx}$ can be represented in the
form (\ref{Eq_Existence_Propagator_deltaprime}). Since the operator $U_{\delta'}(t,x)$ and the Schr\"odinger equation are both linear, we also get
\begin{equation*}
\Psi_n(t,x)=U_{\delta'}(t,x)F_n(\xi)\big|_{\xi=0}
\end{equation*}
for every solution $\Psi_n$ of (\ref{Eq_Schroedinger_deltaprime}) with initial datum $F_n$ of the form (\ref{Eq_Generalized_Fourier_series}).

\vspace{0.2cm}

Let now $(F_n)_n$ be the superoscillating sequence from (\ref{Eq_Convergence_initial_datum}). Then, from the definition of the $A_1$-convergence (\ref{Eq_A1_convergence}), we get the estimate
\begin{equation*}
\left|F_n(\xi)-e^{ik\xi}\right|\leq A_ne^{B|\xi|},\qquad\xi\in\mathbb{C},
\end{equation*}
where $B\geq 0$ is as in (\ref{Eq_A1_convergence}) and $A_n=\sup_{\xi\in\mathbb C}\vert (F_n(\xi)-e^{ik\xi})e^{-B\vert\xi\vert}\vert\rightarrow 0$ for $n\rightarrow\infty$.
Consequently, from the estimate (\ref{Eq_Estimate_propagator_deltaprime}) we get for any
compact $K\subseteq(0,\infty)\times\mathbb{R}$ the uniform convergence
\begin{align*}
\sup\limits_{(t,x)\in K}\big|\Psi_n(t,x)-\Psi(t,x)\big|&=\sup\limits_{(t,x)\in K}\big|U_{\delta'}(t,x)(F_n(\xi)-e^{ik\xi})|_{\xi=0}\big| \\
&\leq\sup\limits_{(t,x)\in K}A_n\widetilde{S}_{eB,\delta'}(t,x)e^{eB|\xi|}\big|_{\xi=0} \\
&=\sup\limits_{(t,x)\in K} A_n\widetilde{S}_{eB,\delta'}(t,x)\overset{n\rightarrow\infty}{\longrightarrow}0;
\end{align*}
note that $\sup_{(t,x)\in K}\widetilde{S}_{eB,\delta'}(t,x)$ is finite, which follows from the continuity
of $\widetilde{S}_{eB,\delta'}$; cf. Lemma~\ref{lem_Existence_Propagator_deltaprime}. This completes the proof of Theorem~\ref{satz_Convergence_solutions}.

\subsection*{Conflict of interest statement}
 On behalf of all authors, the corresponding author states that there is no conflict of interest.


\begin{thebibliography}{99}

\bibitem{AAV88} Y. Aharonov, D. Albert, L. Vaidman, \textit{How the result of a measurement of a component of the spin of a spin-1/2 particle can turn out to be 100}, Phys. Rev. Lett., \textbf{60} (1988), 1351--1354.

\bibitem{ACNSST12} Y. Aharonov, F. Colombo, S. Nussinov, I. Sabadini, D.\,C. Struppa, J. Tollaksen, \textit{Superoscillation phenomena in $SO(3)$}, Proc. Royal Soc. A., \textbf{468} (2012), 3587--3600.

\bibitem{ACSST11} Y. Aharonov, F. Colombo, I. Sabadini, D.\,C. Struppa, J. Tollaksen, \textit{Some mathematical properties of superoscillations}, J. Phys. A, \textbf{44} (2011), 365304, 16pp.

\bibitem{ACSST13} Y. Aharonov, F. Colombo, I. Sabadini, D.\,C. Struppa, J. Tollaksen, \textit{On the Cauchy problem for the Schr\"odinger equation with superoscillatory initial data}, J. Math. Pures Appl., \textbf{99} (2013), 165--173.

\bibitem{ACSST15} Y. Aharonov, F. Colombo, I. Sabadini, D.\,C. Struppa, J. Tollaksen, \textit{Superoscillating sequences as solutions of generalized Schr\"odinger equations}, J. Math. Pures Appl., \textbf{103} (2015), 522--534.

\bibitem{ACSST16} Y. Aharonov, F. Colombo, I. Sabadini, D.\,C. Struppa, J. Tollaksen, \textit{Superoscillating sequences in several variables}, J. Fourier Anal. Appl., \textbf{22} (2016),  751--767.

\bibitem{ACSST17_1} Y. Aharonov, F. Colombo, I. Sabadini, D.\,C. Struppa, J. Tollaksen, \textit{The mathematics of superoscillations}, Mem. Amer. Math. Soc., \textbf{247} (2017), 1174, 107pp.

\bibitem{ACSST17_2} Y. Aharonov, F. Colombo, I. Sabadini, D.\,C. Struppa, J. Tollaksen, \textit{Evolution of superoscillatory initial data in several variables in uniform electric field}, J. Phys. A, \textbf{50} (2017), 185201, 19pp.

\bibitem{ACST18} Y. Aharonov, F. Colombo, D.\,C. Struppa, J. Tollaksen, \textit{Schr\"odinger evolution of superoscillations under different potentials}, Quantum Stud. Math. Found., \textbf{5} (2018), 485--504.

\bibitem{AR05} Y. Aharonov, D. Rohrlich, \textit{Quantum Paradoxes: Quantum Theory for the Perplexed}, Wiley-VCH Verlag, Weinheim, 2005.

\bibitem{ASTY18} Y. Aharonov, I. Sabadini, J. Tollaksen, A. Yger, \textit{Classes of superoscillating functions}, Quantum Stud. Math. Found., \textbf{5} (2018), 439--454.

\bibitem{AV90} Y. Aharonov, L. Vaidman, \textit{Properties of a quantum system during the time interval between two measurements}, Phys. Rev. A, \textbf{41} (1990), 11--20.

\bibitem{AGHH05} S. Albeverio, F. Gesztesy, R. H{\o}egh-Krohn, H. Holden, \textit{Solvable Models in Quantum Mechanics.} With an Appendix by Pavel Exner. 2nd ed. American Mathematical Society, Chelsea Publishing, Providence, RI, 2005.

\bibitem{ACSS18} T. Aoki, F. Colombo, I. Sabadini, D.\,C. Struppa, \textit{Continuity theorems for a class of convolution operators and applications to superoscillations}, Quantum Stud. Math. Found., \textbf{5} (2018), 463--476.

\bibitem{BCS19} J. Behrndt, F. Colombo, P. Schlosser, \textit{Evolution of Aharonov-Berry superoscillations in Dirac $\delta$-potential}, Quantum Stud. Math. Found., \textbf{6} (2019), 279--293.

\bibitem{BEL14} J. Behrndt, P. Exner, V. Lotoreichik, \textit{Schr\"odinger operators with $\delta$ and $\delta'$-interactions on Lipschitz surfaces and chromatic numbers of associated partitions}, Rev. Math. Phys., \textbf{26} (2014), 1450015, 43pp.

\bibitem{BLL13} J. Behrndt, M. Langer, V. Lotoreichik, \textit{Schr\"odinger operators with $\delta$ and $\delta'$-potentials supported on hypersurfaces}, Ann. Henri Poincar\'{e}, \textbf{14} (2013), 385--423.

\bibitem{B94} M.\,V. Berry, \textit{Faster than Fourier}, in Quantum Coherence and Reality; in celebration of the 60th Birthday of Yakir Aharonov (J.\,S. Anandan and J.\,L. Safko, eds), World Scientific, Singapore, (1994), pp.55--65.

\bibitem{B14} M.\,V. Berry, \textit{Superoscillations, Endfire and Supergain}. Quantum Theory: A Two-Time Success Story: Yakir Aharonov Festschrift (D. Struppa, and J. Tollaksen, eds) Springer, New York, pp. 327--336.

\bibitem{B16} M.\,V. Berry \textit{Representing superoscillations and narrow Gaussians with elementary functions}, Milan J. Math., \textbf{84} (2016), 217--230.

\bibitem{B19} M.\,V. Berry et al, \textit{Roadmap on superoscillations}, Journal in Optics, \textbf{21} (2019), 053002.

\bibitem{BD09} M.\,V. Berry, M.\,R. Dennis, \textit{Natural superoscillations in monochromatic waves in D dimension}, J. Phys. A, \textbf{42} (2009), 022003, 8pp.

\bibitem{BP06} M.\,V. Berry, S. Popescu, \textit{Evolution of quantum superoscillations, and optical superresolution without evanescent waves}, J. Phys. A, \textbf{39} (2006), 6965--6977.

\bibitem{B88} S.\,M. Blinder, \textit{Green's function and propagator for the one-dimensional $delta$-function potential}, Phys. Rev. A, \textbf{37} (1988), 973--976.

\bibitem{BEKS94} J.\,F. Brasche, P. Exner, Y.\,A. Kuperin, P. \v{S}eba, \textit{Schr\"odinger operators with singular interactions}, J. Math. Anal. Appl., \textbf{184} (1994), 112--139.

\bibitem{C09} J. Campbell, \textit{Some exact results for the Schr\"odinger wave equation with a time dependent potential}, J. Phys. A, \textbf{42} (2009), 365212, 7pp.

\bibitem{CGS17} F. Colombo, J. Gantner, D.\,C. Struppa, \textit{Evolution of superoscillations for Schr\"odinger equation in uniform magnetic field}, J. Math. Phys., \textbf{58} (2017), 092103, 17pp.

\bibitem{CSSY19} F. Colombo, I. Sabadini, D.\,C. Struppa, A. Yger, \textit{Superoscillating sequences and supershifts for families of generalized functions}, Preprint (2019).

\bibitem{EKMT14} J. Eckhardt, A. Kostenko, M.\,M. Malamud, G. Teschl, \textit{One-dimensional Schr\"odinger operators with $\delta'$-interactions on Cantor-type sets}, J. Differential Equations, \textbf{257} (2014), 415--449.

\bibitem{EGU19} F. Erman, M. Gadella, H. Uncu. \textit{The propagators for $\delta$ and $\delta'$ potentials with time-dependent strengths}, Preprint.

\bibitem{E08} P. Exner, \textit{Leaky quantum graphs: a review}, in: Analysis on graphs and its applications. Selected papers based on the Isaac Newton Institute for Mathematical Sciences programme, Cambridge, UK, 2007; Proc. Symp. Pure Math., \textbf{77} (2008), 523--564.

\bibitem{EK15} P. Exner, H. Kova\v r\' ik, \textit{Quantum Waveguides}, Springer, Heidelberg, (2015).

\bibitem{ER16} P. Exner, J. Rohleder, \textit{Generalized interactions supported on hypersurfaces}, J. Math. Phys., \textbf{57} (2016) 041507, 23pp.

\bibitem{FK04} P.\,J.\,S.\,G. Ferreira, A. Kempf, \textit{Unusual properties of superoscillating particles}, J. Phys. A, \textbf{37} (2004), 12067--76.

\bibitem{FK06} P.\,J.\,S.\,G. Ferreira, A. Kempf, \textit{Superoscillations: faster than the Nyquist rate}, IEEE trans. Signal. Processing, \textbf{54} (2006), 3732--3740.

\bibitem{LF14_1} P.\,J.\,S.\,G. Ferreira, D.\,G. Lee, \textit{Superoscillations of prescribed amplitude and derivative}, IEEE Trans. Signal Processing, \textbf{62} (2014), 3371--3378.

\bibitem{LF14_2} P.\,J.\,S.\,G. Ferreira, D.\,G. Lee, \textit{Superoscillations with optimal numerical stability}, IEEE Sign. Proc. Letters, \textbf{21} (2014), 1443--1447.

\bibitem{H89} J. Herczy\'nski, \textit{On Schr\"odinger operators with distributional potentials}, J. Operator Theory, \textbf{21} (1989), 273--295.

\bibitem{K18} A. Kempf, \textit{Four aspects of superoscillations}, Quantum Stud. Math. Found., \textbf{5} (2018), 477--484.

\bibitem{KM10} A. Kostenko, M.\,M. Malamud, \textit{1-D Schr\"odinger operators with local point interactions on a discrete set}. J. Differential Equations, \textbf{249} (2010), 253--304.

\bibitem{KM14} A. Kostenko, M.\,M. Malamud, \textit{Spectral theory of semibounded Schr\"odinger operators with $\delta'$-interactions}, Ann. Henri Poincar\`e, \textbf{15} (2014), 501--541.

\bibitem{LO16} V. Lotoreichik, T. Ourmi\` eres-Bonafos, \textit{On the bound states of Schr\"odinger operators with $\delta$-interactions on conical surfaces}, Comm. Part. Diff. Eq., \textbf{41} (2016), 999--1028.

\bibitem{LR15} V. Lotoreichik, J. Rohleder, \textit{An eigenvalue inequality for Schr\"odinger operators with $\delta$ and $\delta'$-interactions supported on hypersurfaces}, Oper. Theor. Adv. Appl., \textbf{247} (2015), 173--184.

\bibitem{MPS16} A. Mantile, A. Posilicano, M. Sini, \textit{Self-adjoint elliptic operators with boundary conditions on not closed hypersurfaces}, J. Differential Equations, \textbf{261} (2016), 1--55.

\bibitem{RRSYZ14} E.\,T.\,F. Rogers, T. Roy, Z.\,X. Shen, G.\,H. Yuan, N.\,I. Zheludev,  \textit{Flat super-oscillatory lens for heat-assisted magneticrecording with sub--50 nm resolution}, Opt Express (2014) 226428--37.

\bibitem{RYZ17} E.\,T.\,F. Rogers, G.\,H. Yuan, N.\,I. Zheludev, \textit{Achromaticsuper--oscillatory lenses with subwavelength focusing}, Light Sci. Appl., (2017) 6 e17036.

\bibitem{T52} G. Toraldo di Francia, \textit{Super-Gain Antennas and Optical Resolving Power}, Nuovo Cimento Suppl., \textbf{9} (1952), 426--438.




\end{thebibliography}
\end{document}